\documentclass{svproc} 
   \usepackage{url}

\usepackage{mathptmx}   
\usepackage{helvet} 
\usepackage{courier} 
\usepackage{amsmath,amsfonts}  
\usepackage{graphicx}
\usepackage[bottom]{footmisc}           
    
\usepackage{mathptmx}         
\usepackage{helvet}            
\usepackage{courier}            
\usepackage{type1cm}             
\usepackage{makeidx}             
\usepackage{graphicx}          
\usepackage{multicol}          
\usepackage[bottom]{footmisc}  
     
\usepackage{epsfig}     
\usepackage{psfrag,rotating}     
\usepackage{amssymb,floatflt,enumerate} 
\usepackage{amsmath,amscd,psfrag,leqno} 
\usepackage[mathscr]{eucal}  
\usepackage{shadethm}           
\usepackage{graphicx,subfigure}   
\usepackage{overpic,contour}

\contourlength{0.3mm}  
 
\def\Beweisende{\square}            
\def\BewEnde{\hfill{\Beweisende}}

\def\phi{\varphi}

\def\RR{{\mathbb R}}

\def\EE{{\mathbb E}}
\def\CC{{\mathbb C}}

\def\dach#1{\widehat{#1}}

\def\Vkt#1{{\mathbf #1}}

\newcommand{\mVkt}[1]{\dach{\Vkt #1}}

\newcommand{\kVkt}[1]{\widetilde{\Vkt #1}}

\definecolor{blau}{rgb}{0,0,1}
\newcommand{\blau}{\color{blau}}
\definecolor{rot}{rgb}{1,0,0}
\newcommand{\rot}{\color{rot}}
\definecolor{green}{rgb}{0,1,0} 
\newcommand{\green}{\color{green}}
\definecolor{cyan}{rgb}{0,1,1}
\newcommand{\cyan}{\color{cyan}}
\definecolor{violet}{rgb}{1,.54,1}

\definecolor{magenta}{rgb}{1,0,1}
\newcommand{\magenta}{\color{magenta}}

\begin{document}
\mainmatter              
\title{On the snappability and singularity-distance 
of frameworks with bars and triangular plates}
\titlerunning{Snappability and singularity-distance of frameworks with bars and triangular plates}  
%
\author{G. Nawratil}
\authorrunning{G. Nawratil} 
%
%
\institute{Institute of Discrete Mathematics and Geometry \& \\ Center for Geometry and Computational Design, TU Wien, Austria,\\
\email{nawratil@geometrie.tuwien.ac.at},\\ WWW home page:
\texttt{https://www.dmg.tuwien.ac.at/nawratil/}
}

\maketitle              

\begin{abstract}
In a recent article the author presented a method to measure the 
snapping capability -- shortly called {\it snappability} -- of bar-joint frameworks 
based on the total elastic strain energy by computing the deformation of all bars using 
Hooke's law and the definition of Cauchy/Engineering strain. 
Within the paper at hand, we extend this approach to isostatic frameworks composed of bars and 
triangular plates by using the physical concept of Green-Lagrange strain. 
An intrinsic pseudometric based on the resulting total elastic strain energy density 
cannot only be used for evaluating the snappability but also for   
measuring the distance to the closest singular configuration. 
The presented methods are demonstrated on the basis of the 3-legged planar parallel manipulator.
\end{abstract}

\section{Introduction}
\label{sec:introduction}
A framework in the Euclidean space $\EE^n$ consists of a knot set $\mathcal{K}=\left\{K_1,K_2,\ldots,K_s\right\}$ 
and an abstract graph $G$ on $\mathcal{K}$ fixing the combinatorial structure. We denote the edge connecting  $K_i$ to $K_j$ 
by $e_{ij}$ with $i<j$ and collect all indices of knots edge-connected to $K_i$ in the knot neighborhood $N_i$.  
Moreover we denote the number of edges in the graph by $b$ and fix the intrinsic metric of the framework by assigning a length $L_{ij}\in\RR_{>0}$ to each edge $e_{ij}$. 
In general this assignment does not determine the shape of the framework uniquely 
thus a framework has different incongruent 
realizations. For example, a triangular framework has in general two realizations in  $\EE^2$, which are not congruent with respect to the group of direct isometries. 
If we consider the isometry group  
then this number halves. 
We denote a framework's realization by $G(\Vkt k)$ where the configuration of knots $\Vkt k:=(\Vkt k_1,\ldots ,\Vkt k_s)$ is composed of the $n$-dimensional
coordinate vectors $\Vkt k_i$ of the knots $K_i$ ($i=1,\ldots, s$). 
Note that a framework is called isostatic if the removal of any edge of $G$ results in a flexible framework.

In general we materialize edges $e_{ij}$ by straight bars, but if three edges $e_{ij}$, $e_{ik}$ and $e_{jk}$ form a triangle structure then 
the three bars can alternatively be replaced by a triangular plate\footnote{A $r$-plate is a compact connected set in $\EE^n$ whose affine span is $r$-dimensional 
according to \cite{kiraly}. Triangular plates refer to $2$-plates of triangular shape.}.  
The elements of the framework are linked in the planar case ($n=2$) by rotational joints and in the spatial case ($n=3$) either by spherical joints or hinges. 
We assume that 
\begin{enumerate}[(I)]
\item
all bars and triangular plates are uniform made of the same homogeneous isotropic material deforming at constant volume,
\item
all bars have the same cross-sectional area $A$,
\item
all joints are without clearance.
\end{enumerate}

A realization  
is called a snapping realization if it is {\it close enough} 
to another incongruent realization such that the physical model can snap into this neighboring realization due to non-destructive elastic deformations of material. 
Shakiness 
can be seen as the limiting case where two realizations of a framework coincide; e.g.\ \cite{wohlhart}. 

We define an intrinsic pseudometric based on the total elastic strain energy density of the isostatic  
framework (Sec.\ \ref{totalGL}) using the physical model of Green-Lagrange (GL) strain (Sec.\ \ref{basicsGL}). This metric is 
then employed to measure (a) the snapping capability (shortly called {\it snappability}) of a realization (Sec.\ \ref{sec:snap}) 
and (b) the distance to the next shaky (also referred to as {\it singular} or {\it infinitesimal flexible}) realization (Sec.\ \ref{sec:singdist}).

One can apply the proposed approach to almost all known examples\footnote{For a detailed review please see \cite{ark2020} and the references therein.} 
of snapping spatial frameworks (cf.\ \cite{ijss}), as for example the {\it Siamese dipyramids} \cite{goldberg,gorkavyy}, 
the {\it four-horn} \cite{schwabe} or Wunderlich's snapping octahedra, icosahedra and dodecahedra, which are reviewed in \cite{stachel_wunderlich}.
But the presented method is not limited to the listed triangular plate-hinge structures -- also known as panel-hinge 
frameworks\footnote{A body-hinge framework with the property that all hinges of each body are coplanar \cite{kiraly}.} --  
as it can also handle structures including bars, as for example the 3-legged planar parallel 
manipulator or an spatial hexapod of octahedral structure, which are both of practical importance. 
Especially for these mechanical devices also the proposed singularity-distance is of interest (e.g.\ for path planning), 
which can be seen as an alternative to the extrinsic metrics presented by the author in \cite{WC_2019}.
Therefore we demonstrate our methods on the basis of a 3-legged planar parallel manipulator (cf.\ Ex.\ \ref{ex1}).

\section{Elastic GL strain energy of bars and triangular plates}\label{basicsGL}

In \cite{ark2020} the elastic strain energy stored in a deformed bar $e_{ij}$ 
was computed by 
\begin{equation}\label{cauchy}
U_{ij}=\frac{EA}{2L_{ij}} (L_{ij}'-L_{ij})^2
\end{equation}
where $E$ denotes the modulus of 
elasticity\footnote{In this paper we assume $E>0$ as for conventional structural material $E$ is positive.}, 
 $L_{ij}'$ is the deformed length of the bar and $L_{ij}$ its original one.  
The formula (\ref{cauchy}) is based on the so-called Cauchy/Engineering (CE) strain, which 
can also be extended to triangular elements playing a central role in the {\it plane stress}\footnote{The shear stress and normal 
stress perpendicular to the plane of the triangle is zero.}
analysis within the finite element method  (e.g.\ see \cite[Chapter 6]{logan}). But the resulting elastic strain energy of a triangular plate is not  
invariant under rotations; i.e.\ a pure rotation already implies a deformation energy. 
Therefore this formulation is not suited for kinematic considerations.
As a consequence we follow a more sophisticated approach; namely the GL strain (e.g.\ see \cite[Sec.\ 2.4.2]{reddy}), 
which is summarized in the remainder of this section.

Let $K_i, K_j, K_k$ denote the vertices of the triangular plate in the given undeformed configuration 
and  $K'_i, K'_j, K'_k$ in the deformed one. 
Then there exists a uniquely defined $2\times 2$ matrix $\Vkt A$ which has the property
\begin{equation}\label{def:affine}
\Vkt A(\mVkt k_j-\mVkt k_i)=\mVkt k'_j-\mVkt k'_i, \qquad 
\Vkt A(\mVkt k_k-\mVkt k_i)=\mVkt k'_k-\mVkt k'_i, 
\end{equation}
where $\mVkt k_z$ (resp.\ $\mVkt k'_z$) is a 2-dimensional vector of $K_z$ (resp.\  $K'_z$) for $z\in\left\{i,j,k\right\}$ 
with respect to a planar Cartesian frame $\mathcal F$ (resp.\  $\mathcal F'$) attached to the carrier plane of the triangle $K_i, K_j, K_k$ (resp.\ $K'_i, K'_j, K'_k$). 
Then the GL normal strains $\varepsilon_{x}$ and $\varepsilon_{y}$, respectively, and the 
GL shear strain $\gamma_{xy}$ can be computed as 
\begin{equation}
\begin{pmatrix}
\varepsilon_{x} & \gamma_{xy} \\
\gamma_{xy} & \varepsilon_{y}
\end{pmatrix}=\frac{1}{2}\left(\Vkt A^T\Vkt A-\Vkt I\right).
\end{equation}
We reassemble these quantities in the vector $\Vkt e=(\varepsilon_{x},\varepsilon_{y},2\gamma_{xy})^T$. 
Using this notation the elastic GL strain energy of the deformation can be calculated as 
\begin{equation}\label{esetp}
U_{ijk}=V_{ijk}\tfrac{1}{2}\Vkt e^T\Vkt D\Vkt e
\end{equation}
where $V_{ijk}$ denotes the volume of the triangular plate and $\Vkt D$ the planar stress/strain matrix (constitutive matrix), which reads as:  
\begin{equation}
\Vkt D=\frac{E}{1-\nu^2}
\begin{pmatrix}
1 & \nu & 0 \\
\nu & 1 & 0 \\
 0 & 0 & \tfrac{1-\nu}{2}
\end{pmatrix}.
\end{equation}
We can set the Poisson’s ratio $\nu$ to one-half due to   
the assumed invariance of the volume $V_{ijk}$ under deformation (cf.\ assumption I), which is used later on for the computation of the total elastic GL strain energy density. 
Moreover we set $E=1$ as done in \cite{ark2020} in order to reduce the physical formulation to its geometric core. 

Following the same approach the elastic GL strain energy of a deformed bar can be computed as 
\begin{equation}\label{bar_GL}
U_{ij}=\frac{A}{8L_{ij}^3} (L'_{ij}-L_{ij})^2(L'_{ij}+L_{ij})^2. 
\end{equation}


\section{The framework's total elastic GL strain energy and its density}\label{totalGL}

The total  elastic GL strain energy $U$ of a framework composed of bars and triangular plates 
results from the summation of the plate energies (\ref{esetp}) and the bar energies (\ref{bar_GL}).

Recall that we can model a triangular structure either as bar-joint framework or as triangular plate.  
In order to ensure a fair comparability of both approaches, the used amount of material has to be the same; i.e.\
$V_{ijk}=A(L_{ij}+L_{ik}+L_{jk})$.
Taking this relation into account the following lemma holds:

\begin{lemma}\label{basic}
The total elastic GL strain energy $U$ of a framework composed of bars and triangular plates 
is a rational polynomial function with respect to the intrinsic metric of the framework. 
The polynomial in the denominator only depends on the undeformed edge lengths $L_{ij}$. The polynomial in  
the numerator also includes the deformed edge lengths  $L_{ij}'$ and it is of degree 4 with 
respect to these variables $L_{ij}'$ which only appear with even powers. 
\end{lemma}

\begin{proof}
We choose the planar Cartesian frame $\mathcal F$ in a way that 
its origin equals $K_i$ and that $K_j$ is located on its positive $x$-axis; i.e.\  
$\mVkt k_i=(0,0)^T$, $\mVkt k_j=(L_{ij},0)^T$ and
\begin{equation}
\mVkt k_k=\left(\tfrac{L_{ij}^2+L_{ik}^2-L_{jk}^2}{2L_{ij}}, 
\pm\tfrac{\sqrt{
(L_{ij} + L_{ik} + L_{jk})(L_{ij} - L_{ik} + L_{jk})(L_{ij} + L_{ik} - L_{jk})(L_{ik} + L_{jk} - L_{ij})}}{2L_{ij}}
\right)^T
\end{equation}
where the $y$-coordinate can have positive or negative sign for the dimension $n=2$ 
depending on the orientation of the triangle $K_i, K_j, K_k$.  For the dimension 
 $n>2$ one can always assume a positive sign.   
Similar considerations can be done for the planar Cartesian frame $\mathcal F'$ with respect to the triangle $K_i', K_j', K_k'$ 
where one ends up with exactly the same coordinatisation as above but only primed. Inserting these coordinates of the 
six vectors $\mVkt k_i,\mVkt k_j,\mVkt k_k,\mVkt k_i',\mVkt k_j',\mVkt k_k'$ into Eq.\ (\ref{esetp})
shows the result\footnote{The obtained expression is independent of the 
sign of the $y$-coordinate of $\mVkt k_k$ and $\mVkt k_k'$.}
 for plates.  For bars this result is directly visible from Eq.\ (\ref{bar_GL}), which concludes the proof. 
\hfill $\BewEnde$
\end{proof}

As in our case the undeformed lengths $L_{ij}$ are given we can interpret $U$ as a function of the bar lengths $\Vkt L' = (\ldots, L_{ij}' ,\ldots)^T\in \RR^b$ 
of a realization $G(\Vkt k')$; i.e.\ $U(\Vkt L')$. 
Note that due to Lemma \ref{basic} the formula for $U(\Vkt L')$ can be written in matrix formulation as 
$U(\Vkt L')=\kVkt L'^T \Vkt M \kVkt L'$ where $\Vkt M$ is a symmetric $(b+1)$-matrix and $\kVkt L':=(1, \ldots ,L_{ij}'^2, \ldots)^T$
is composed of the $b$ squared edge lengths and the number $1$.

From the underlying physical interpretation it seems to be clear that the elastic strain energy $U_{ijk}$ given in Eq.\ (\ref{esetp}) is positive semi-definite; but one can also prove 
this mathematically by decomposing it into a sum of squares (see e.g.\ \cite{reznick}). 
For $U_{ij}$ given in Eq.\ (\ref{bar_GL}) it can immediately be seen that it is positive semi-definite.
Due to the resulting positive semi-definiteness of $U(\Vkt L')$ its density $D(\Vkt L'):=U(\Vkt L')/(AL)$,  where $L$ is the total length $L=\sum_{i<j}L_{ij}$ of the framework, 
can be used for building up the following intrinsic pseudometric $d_p$ for framework realizations: 
\begin{equation}\label{pseudo}
d_p:\,\,
\RR^b\times\RR^b\rightarrow \RR_{\geq 0} \quad\text{with}\quad
(\Vkt L',\Vkt L'')\mapsto |D(\Vkt L')-D(\Vkt L'')|
\end{equation}
where $\Vkt L'' = (\ldots, L_{ij}'' ,\ldots)^T\in \RR^b$ collects the bar lengths of another realization $G(\Vkt k'')$.


\subsection{Snappability}\label{sec:snap}

\begin{theorem}\label{thm:critic}
The critical points of the total elastic GL strain energy $U(\Vkt k')$ of an isostatic framework correspond to realizations $G(\Vkt k')$ that are either undeformed or 
deformed and  shaky. 
\end{theorem}

\begin{proof}
The proof is based on the following characterization of shakiness in terms of self-stress (e.g.\ \cite{connelly_book}):  
If one can assign to each edge $e_{ij}$ of  $G(\Vkt k')$ a  {\it stress} ${\omega}_{ij}\in\RR$ in a way that for each 
knot the so-called {\it equilibrium condition}
\begin{equation}\label{equilibrium}
\sum_{i<j\in N_i}{\omega}_{ij}(\Vkt k'_i-\Vkt k'_j) + \sum_{i>j\in N_i}{\omega}_{ji}(\Vkt k'_i-\Vkt k'_j)=\Vkt o
\end{equation}
is fulfilled, where $\Vkt o$ denotes the $n$-dimensional zero vector, then the $b$-dimensional vector ${\omega}=(\ldots, {\omega}_{ij} , \ldots)^T$ is referred as 
{\it self-stress}. 
If ${\omega}$ differs from the zero vector, then the realization $G(\Vkt k')$ of an isostatic framework is shaky (e.g.\ \cite{gluck,roth}). 

Now we consider $U$ in dependence of the configuration of knots $\Vkt k'$, i.e.\ $U(\Vkt k')$ and compute 
the system of equations characterizing its critical points as
\begin{equation}
\nabla_{\hspace{-0.5mm}i}\, U(\Vkt k')=\Vkt o  \quad \text{with} 
\quad
\nabla_{\hspace{-0.5mm}i}\, U(\Vkt k')=\left(\tfrac{\partial U}{\partial k'_{i,1}}, 
 \ldots , 
\tfrac{\partial U}{\partial k'_{i,n}}\right) \quad \text{and} 
\quad
i=1,\ldots, s
\end{equation}
where $(k'_{i,1},\ldots , k'_{i,n})$ is the coordinate vector of $\Vkt k'_i$.
Due to the sum rule for derivatives we only have to investigate $\nabla_{\hspace{-0.5mm}i}$ 
of $U_{ijk}(\Vkt k')$ and  $U_{ij}(\Vkt k')$ given in Eqs.\ (\ref{esetp}) and (\ref{bar_GL}). 
\begin{enumerate}[1.]
\item Due to  
$\nabla_{\hspace{-0.5mm}i}\, U_{ij}(\Vkt k') = \tfrac{A(L_{ij}'^2-L_{ij}^2)}{2L_{ij}^3}(\Vkt k'_i-\Vkt k'_j)$ Theorem \ref{thm:critic} is valid for frameworks, which only consist of bars, 
as $\nabla_{\hspace{-0.5mm}i}\, U(\Vkt k')$ equals 
Eq.\ (\ref{equilibrium}) with ${\omega}_{ij}={A(L_{ij}'^2-L_{ij}^2)}/(2L_{ij}^3)$. 
\item
If triangular plates are involved we consider  
$\nabla_{\hspace{-0.5mm}i}\, U_{ijk}(\Vkt k')$, $\nabla_{\hspace{-0.5mm}j}\, U_{ijk}(\Vkt k')$ and $\nabla_{\hspace{-0.5mm}k}\, U_{ijk}(\Vkt k')$. 
Straight forward symbolic computations (e.g.\ using Maple) show that the following overdetermined system of equations
\begin{equation}\label{test}
\begin{split}
{\omega}_{ij}(\Vkt k'_i-\Vkt k'_j) + {\omega}_{ik}(\Vkt k'_i-\Vkt k'_k) - \nabla_{\hspace{-0.5mm}i}\, U_{ijk}(\Vkt k') &=\Vkt o \\
{\omega}_{ij}(\Vkt k'_j-\Vkt k'_i) + {\omega}_{jk}(\Vkt k'_j-\Vkt k'_k) - \nabla_{\hspace{-0.5mm}j}\, U_{ijk}(\Vkt k') &=\Vkt o \\
{\omega}_{ik}(\Vkt k'_k-\Vkt k'_i) + {\omega}_{jk}(\Vkt k'_k-\Vkt k'_j) - \nabla_{\hspace{-0.5mm}k}\, U_{ijk}(\Vkt k') &=\Vkt o 
\end{split}
\end{equation}
has a unique solution for ${\omega}_{ij}$, ${\omega}_{ik}$ and ${\omega}_{jk}$ if $K_i', K_j', K_k'$ generate a triangle. If these points are collinear we 
get a positive dimensional solution set. Hence, one can replace $\nabla_{\hspace{-0.5mm}i}\, U_{ijk}(\Vkt k')$ by a linear combination
${\omega}_{ij}(\Vkt k'_i-\Vkt k'_j) + {\omega}_{ik}(\Vkt k'_i-\Vkt k'_k)$ where the coefficients ${\omega}_{ij}$ and 
${\omega}_{ik}$ are compatible with the other equations of (\ref{test}). As a consequence $\nabla_{\hspace{-0.5mm}i}\, U(\Vkt k')$ can again be written in 
the form of Eq.\ (\ref{equilibrium}).  
\hfill $\BewEnde$
\end{enumerate}
\end{proof}

This result implies that the Theorems 1 and 2 of \cite{ark2020}\footnote{The isostaticity was tacitly assumed for the results stated in  \cite{ark2020} 
by referring to \cite{gluck,roth}.}
also hold true for isostatic frameworks with bars and triangular plates. 
They can be summed up as follows: 

\begin{theorem}\label{thm12}
If an isostatic framework snaps out of a stable\footnote{A realization $G(\Vkt k)$ is called stable if it corresponds to local minimum of $U(\Vkt k)$.} realization $G(\Vkt k)$ 
by applying the minimum GL strain energy needed to it, then 
the corresponding deformation of the realization has to pass a shaky realization $G(\Vkt k')$ at the maximum state of deformation. 
Such a snap of a framework ends up in a realization $G(\Vkt k'')$ which is either undeformed or deformed and shaky. 
\end{theorem}

Therefore the snappability $s(\Vkt k)$ of an undeformed realization $G(\Vkt k)$ can be measured by $d_p(\Vkt L',\Vkt L)=D(\Vkt L')$ of Eq.\ (\ref{pseudo}). 
In the following we present the procedure for determining $G(\Vkt k')$, which is similar to the one given in \cite{ark2020}.
As preparatory work for this algorithm we define the quotient set 
$\mathcal{R}:=\mathcal{S}/SE(n)$ where 
$SE(n)$ denotes the group of direct isometries of $\EE^n$ and
$\mathcal{S}$ the set of real saddle points of $U(\Vkt k')$, which can be selected from 
the critical points via the second derivative test. 

Let us assume that  $G(\Vkt k')\in\mathcal{R}$ yields the minimal 
value for $d_p(\Vkt L',\Vkt L)=D(\Vkt L')$ where $G(\Vkt k)$ is the given undeformed realization. The following equation 
\begin{equation}\label{imply}
\kVkt L_t:=\kVkt L+t(\kVkt L'-\kVkt L) \quad \text{with}\quad t\in[0,1]
\end{equation}
implies a path $\Vkt L_t$ in $\RR^b$ between $\Vkt L$ and $\Vkt L'$. 
Along this path  
the deformation energy $U_{ijk}$ of each triangular plate and the deformation energy $U_{ij}$ 
of each bar is {\it monotonic increasing} with respect to the curve parameter $t$, which  
ensures that the minimum mechanical work needed is applied on the framework to reach $G(\Vkt k')$. 
This results from Lemma \ref{basic}, as $U_{ijk}(\Vkt L_t)$ as well as $U_{ij}(\Vkt L_t)$  are quadratic functions in $t$, which are at their minima for $t=0$. 
The path $\Vkt L_t$ corresponds to different 1-parametric deformations  of realizations in $\EE^n$. 
If among these a deformation $G(\Vkt k_t)$ with the property 
\begin{equation}\label{property}
G(\Vkt k_t)\big|_{t = 0}=G(\Vkt k),\quad G(\Vkt k_t)\big|_{t = 1}=G(\Vkt k') 
\end{equation}
exists, then the given realization $G(\Vkt k)$ is deformed into $G(\Vkt k')$ under $\Vkt L_t$. 
Computationally the property (\ref{property}) can be checked by a user defined homotopy approach 
relying on the software Bertini \cite[Sec.\ 2.3]{bates}. 
If such a deformation does not exist then we redefine 
$\mathcal{R}$ as $\mathcal{R}\setminus\left\{ G(\Vkt k')\right\}$ and run again the procedure explained in this paragraph until we 
get the sought-after realization $G(\Vkt k')$ implying $s(\Vkt k)$ by $U(\Vkt k')$.
If we end up with $\mathcal{R}=\varnothing$ then we set $s(\Vkt k)=\infty$.


\subsection{Singularity-distance}\label{sec:singdist}

One can rewrite the $s$ equations given in (\ref{equilibrium})
in matrix form as $\Vkt R_{G(\Vkt k')}{\omega}=\Vkt o$, where the $(sn\times b)$-matrix $\Vkt R_{G(\Vkt k')}$ is the so-called {\it rigidity matrix}.
It is well known (e.g.\ \cite{meera}) that a realization $G(\Vkt k')$ is shaky  
if and only if $rk(\Vkt R_{G(\Vkt k')})<r$ with $r:=sn-(n^2+n)/2$. Clearly, based on this rank condition one can characterize all shaky realizations $G(\Vkt k')$ algebraically by  
the variety $V(J)$ where $J$ denotes the ideal generated by all minors of $\Vkt R_{G(\Vkt k')}$ of order $r\times r$. 
For isostatic frameworks this results in a single condition, which is also known as {\it pure condition} and can be 
derived in terms of brackets (cf.\ \cite{white}).

In the following we want to determine the real point $\Vkt k'$ of this {\it shakiness variety} $V(J)$ which minimizes the 
value $d_p(\Vkt L',\Vkt L)=D(\Vkt L')$,  
where $G(\Vkt k)$ denotes the given undeformed realization. Moreover there should again exist a 1-parametric deformation $\Vkt L_t$ implied by Eq.\ (\ref{imply}) 
such that the properties of Eq.\ (\ref{property}) hold. If this is the case then we call $d_p(\Vkt L',\Vkt L)=D(\Vkt L')$ the singularity-distance $\sigma(\Vkt k)$.

\begin{theorem}\label{thm:ident}
For a realization $G(\Vkt k)$ of an isostatic framework, which is undeformed and not shaky,  the singularity-distance $\sigma(\Vkt k)$ equals the snappability $s(\Vkt k)$. 
\end{theorem}

\begin{proof}
$\sigma(\Vkt k)\leq s(\Vkt k)$ has to hold, as
the realization $G(\Vkt k')$ of Theorem \ref{thm12}, which implies the snappability $s(\Vkt k)$, is shaky. 
We show the equality indirectly by assuming $\sigma(\Vkt k)< s(\Vkt k)$. We denote the shaky realization implying $\sigma(\Vkt k)$ 
by $G(\Vkt k'')$ which  corresponds to $\Vkt L''\in\RR^b$. 
In analogy to Eq.\  (\ref{imply}) we consider the relation
\begin{equation}
\kVkt L_t:=\kVkt L+t(\kVkt L''-\kVkt L) \quad \text{with}\quad t\in[0,1]
\end{equation}
defining a path $\Vkt L_t$ in $\RR^b$ between $\Vkt L$ and $\Vkt L''$, which corresponds to a set 
of 1-parametric deformations $\left\{G(\Vkt k_t^1),G(\Vkt k_t^2), \ldots \right\}$. 
A subset $\mathcal{D}$ of this set has the property $G(\Vkt k_t^i)|_{t = 1}=G(\Vkt k'')$ where $\#\mathcal{D}>1$ holds as 
$G(\Vkt k'')$ is shaky \cite{wohlhart,stachel_wunderlich}. 
Therefore the isostatic framework can snap out of $G(\Vkt k)$ over $G(\Vkt k'')$ which contradicts $\sigma(\Vkt k)< s(\Vkt k)$ 
($\Longrightarrow$ $\sigma(\Vkt k)= s(\Vkt k)$). 
\hfill $\BewEnde$
\end{proof}

\begin{figure}[t]
\begin{minipage}{65mm}
\begin{center} 
\begin{overpic}
    [width=65mm]{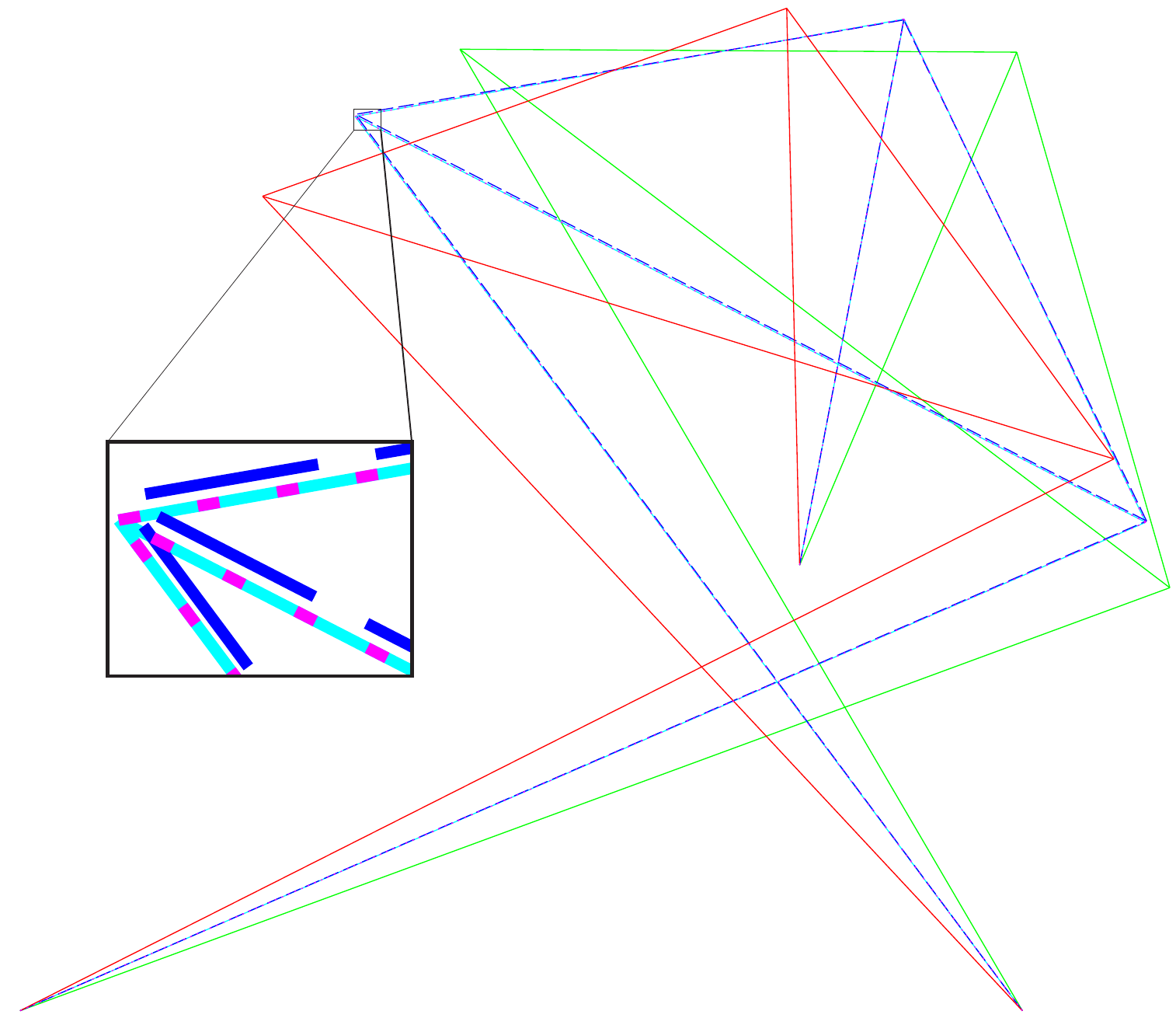}  
\begin{scriptsize}
\put(0.5,3.5){$K_1$}
\put(88.5,0){$K_2$}
\put(69,38){$K_3$}
\end{scriptsize} 		
  \end{overpic} 
\end{center} 
\end{minipage}
\hfill
\begin{scriptsize}
\begin{minipage}{55mm}
\begin{center}
{\bf Computational 
Data}
\begin{align*}
&{\green \Vkt k_4=\begin{pmatrix}10.3238 \\ 3.7970\end{pmatrix}} &\, &{\green \Vkt k_5=\begin{pmatrix}3.9493 \\8.6308\end{pmatrix}} &\, &{\green \Vkt k_6=\begin{pmatrix}8.9493 \\ 8.6043\end{pmatrix}} \\
&{\rot \Vkt k_4=\begin{pmatrix} 9.8218 \\ 4.9528 \end{pmatrix}} &\, &{\rot \Vkt k_5=\begin{pmatrix} 2.1773 \\ 7.3110 \end{pmatrix}} &\, &{\rot \Vkt k_6=\begin{pmatrix} 6.8838 \\ 8.9986 \end{pmatrix}} \\
&{\blau \Vkt k_4=\begin{pmatrix} 10.1168 \\ 4.3957 \end{pmatrix}} &\, &{\blau \Vkt k_5=\begin{pmatrix} 3.0101 \\ 8.0463 \end{pmatrix}} &\, &{\blau \Vkt k_6=\begin{pmatrix} 7.9343 \\ 8.8931 \end{pmatrix}} \\
&{\cyan \Vkt k_4=\begin{pmatrix} 10.1071  \\ 4.3844  \end{pmatrix}} &\, &{\cyan \Vkt k_5=\begin{pmatrix} 3.0084  \\ 8.0281  \end{pmatrix}}  &\, &{\cyan \Vkt k_6=\begin{pmatrix} 7.9382  \\ 8.9006 \end{pmatrix}} \\
&{\magenta \Vkt k_4=\begin{pmatrix} 10.1071 \\  4.3845 \end{pmatrix}} &\, &{\magenta \Vkt k_5=\begin{pmatrix} 3.0084  \\ 8.0282  \end{pmatrix}} &\, &{\magenta \Vkt k_6=\begin{pmatrix} 7.9382 \\ 8.9006 \end{pmatrix}} 
\end{align*}
\centering
  \begin{tabular}{|c|c|c|c|c|} \hline
  \, &   tracked paths & $\#$solutions & $\#\mathcal{R}$  & $\sigma(\Vkt k)=s(\Vkt k)$   \\
    \hline
			{\blau $G(\Vkt k')$}   &   $729$ & $285$ & 62  & $3.2531/10^6$  \\
	 \hline
			{\cyan $G(\Vkt k')$}   &   $729$ & $219$ & $58$  & $1.8271/10^6$  \\
 \hline
			{\magenta $G(\Vkt k')$}   &  $59\,163$ \cite{ark2020} & 758 & 142  & $1.8285/10^6$  \\
 \hline
  \end{tabular}
\end{center}
\end{minipage}
\end{scriptsize}
\caption{3-legged planar parallel manipulator of Ex.\ \ref{ex1}: Illustration (left) and data (right). 
}\label{fig:1}	
\end{figure}

\noindent
\begin{remark} Our theoretical  considerations end with the following three final remarks: 
\begin{enumerate}[(a)]
\item
The deformation $G(\Vkt k_t)$ of Theorem \ref{thm:ident}  implying $\sigma(\Vkt k)=s(\Vkt k)$ has to be real for $t\in[0,1]$, 
as a real solution can only change over into a complex one through a double root, 
which corresponds to a shaky realization ($\Rightarrow$ contradiction to the definition of $\sigma(\Vkt k)$).  
\item
The results of Sec.\ \ref{sec:singdist} also hold if one uses the CE strain approach (cf.\ \cite{ark2020}). 
Therefore item (a) gives an answer to the open problem stated in \cite[Remark 3]{ark2020}, but our 
concept still ignores the collision of bars and/or plates during the deformation.
\item
The results of this paper also hold for pinned frameworks  (cf.\ \cite[Sec.\ 3.3]{ark2020}). \hfill $\diamond$
\end{enumerate}
\end{remark}

\begin{example}\label{ex1}
We consider a 3-legged planar parallel manipulator (cf.\ Fig.\ \ref{fig:1})  with a pinned base given by 
$\Vkt k_1=(0,0)^T$, $\Vkt k_2=(9,0)^T$, $\Vkt k_3=(7,4)^T$ which is equipped with 
the intrinsic metric $(L_{14}, L_{25}, L_{36}, L_{45}, L_{46}, L_{56})=(11,10,5,8,5,5)$.
The two undeformed realizations {\green $G(\Vkt k)$} and {\rot $G(\Vkt k)$} snap into each other over $G(\Vkt k')$ 
which was computed based on a framework consisting of (a) six bars using GL/CE strain ({\cyan cyan}/{\magenta magenta dotted}), (b) three bars and one triangular plate using GL strain ({\blau blue dashed}).
The computation of the critical points of $U(\Vkt k')$ was performed with the software Bertini \cite{bates}. The number of 
tracked paths for each approach is given in the table displayed in Fig.\ \ref{fig:1}-right as well as the number of 
paths ending up in finite solutions (over $\CC$) under the homotopy continuation. Moreover, the cardinal numbers of 
$\mathcal{R}$ and  $\mathcal{Q}$ are given as well as the value for $\sigma(\Vkt k)=s(\Vkt k)$. The coordinates of the corresponding 
configurations $G(\Vkt k')$ are printed above this table. 
These three configurations are very close together, especially the  {\cyan cyan} and {\magenta magenta} one 
(they are identical up to the first three digits after the comma). 
The latter difference is even not visible in the blow-up provided in Fig.\ \ref{fig:1}-left. 
\end{example}

\section{Conclusion}

We presented the computation of the snappability and the singularity-distance of isostatic frameworks composed of 
bars and triangular plates based on the total elastic strain energy density using the physical concept of GL strain. 
This measure enables the fair comparison of frameworks, which differ in the number of knots, the combinatorial structure, the intrinsic metric and 
the realization of triangular structures (triangular plates vs.\ joint-bar triangles). 
Our methods are demonstrated on the basis of the example of a  3-legged planar parallel manipulator,  
which also points out the computational efficiency of the proposed approach compared to the one using the 
definition of CE strain \cite{ark2020}. 

For a more detailed formulation of this approach including its generalization to frameworks 
involving polygonal plates and/or polyhedra please see \cite{ijss}, where also examples of spatial 
snapping structures are discussed with special emphasis on Stewart-Gough manipulators.

\paragraph{{\bf Acknowledgments}}
The research is supported by Grant No.\ P\,30855-N32 of the Austrian Science Fund FWF.


\begin{thebibliography}{99.}

\bibitem{bates}
Bates, D.J., Hauenstein J.D., Sommese, A.J., Wampler C.W.: 
Numerically Solving Polynomial Systems with Bertini. SIAM Philadelphia (2013)

\bibitem{connelly_book}
Connelly, R.: Rigidity. Handbook of Convex Geometry (P.M.\ Gruber, J.M.\ Wills eds.), pages 223--271, Elsevier (1993)

\bibitem{gluck}
Gluck, H.: Almost all simply connected closed surfaces are rigid. 
Geometric Topology (L.C.\ Glaser, T.B. Rushing eds.), pages 225--239, Springer (1975)  

\bibitem{goldberg}
Goldberg, M.:
Unstable Polyhedral Structures. Math.\ Mag.\ \textbf{51}(3) 165--170 (1978)

\bibitem{gorkavyy}
Gorkavyy, V., Fesenko, I.:
On the model flexibility of Siamese dipyramids. J.\ Geom.\ \textbf{110}:7 (2019)

\bibitem{kiraly}
Kiraly, C., Tanigawa, S.: Rigidity of Body-Bar-Hinge Frameworks. Handbook of Geometric Constraint Systems Principles (M.\ Sitharam et al eds.), 
pages 435--459, CRC Press (2019)

\bibitem{logan}
Logan, D.L.: A First Course in the Finite Element Method. 4th Ed., Thomson (2007)

\bibitem{WC_2019}
Nawratil, G.: Singularity Distance for Parallel Manipulators of Stewart Gough Type. Advances in Mechanism and Machine Science 
(T.\ Uhl ed.), pages 259--268, Springer (2019)

\bibitem{ark2020}
Nawratil, G.: Evaluating the snappability of bar-joint frameworks. 
Advances in Robot Kinematics 2020 (J.\ Lenarcic, B.\ Siciliano eds.), pages 182--189, Springer (2020)

\bibitem{ijss}
Nawratil, G.: Snappability and singularity-distance of pin-jointed body-bar frameworks. arXiv:2101.02490 (2021) 

\bibitem{reddy}
Reddy, J.N.: An Introduction to Nonlinear Finite Element Analysis. 2nd Ed., Oxford University Press (2015)

\bibitem{reznick}
Reznick, B.: Some concrete aspects of Hilbert's 17th Problem. Real Algebraic Geometry and Ordered Structures (C.N.\ Delzell, J.J.\ Madden eds.),
pages 251--272, AMS (2000) 

\bibitem{roth}
Roth, B.: 
Rigid and Flexible Frameworks. Amer.\ Math.\ Monthly \textbf{88}(1) 6--21 (1981)

\bibitem{meera}
Sitharam, M., Baker, T.: Overview and Preliminaries. Handbook of Geometric Constraint Systems Principles (M.\ Sitharam et al eds.), 
pages 1--17, CRC Press (2019)


\bibitem{stachel_wunderlich}
Stachel, H.: W.\ {W}underlichs {B}eitr\"age zur {W}ackeligkeit. 
Technical Report No.\ 22, Institute of Geometry, TU Wien (1995)


\bibitem{white}
White, N.L., Whiteley, W.: The algebraic geometry of stresses in frameworks, SIAM J.\ Alg.\ Disc.\ Meth.\ \textbf{4}(4) 481--511 (1983)

\bibitem{wohlhart}
Wohlhart, K.: Degrees of shakiness. 
Mech.\ Mach.\ Theory \textbf{34}(7) 1103--1126 (1999)

\bibitem{schwabe}
Wunderlich, W., Schwabe, C.: 
Eine {F}amilie von geschlossenen gleichfl\"achigen {P}olyedern, die fast beweglich sind. 
Elem. Math.\ \textbf{41}(4) 88--93 (1986)

\end{thebibliography}
\end{document}